\documentclass[12pt]{article}
\usepackage{amsmath}
\usepackage{amssymb}
\usepackage{amsthm}
\usepackage{fullpage}
\usepackage[utf8]{inputenc}
\usepackage{mathtools}
\usepackage{url}
\usepackage[usenames,svgnames]{xcolor}
\usepackage{cite}

\usepackage[pagebackref=false]{hyperref} 
\usepackage[all]{hypcap} 

\theoremstyle{plain}
\newtheorem{theorem}{Theorem}[section]

\theoremstyle{definition}

\newtheorem{remark}[theorem]{Remark}

\numberwithin{equation}{section} 

\hypersetup{
  breaklinks,colorlinks,
  citecolor=Green,
  linkcolor=BlueViolet,
  urlcolor=DarkBlue,
  pdftitle={2D Toda tau-functions as combinatorial generating functions},
  pdfauthor={Mathieu Guay-Paquet and John Harnad},
}

\DeclareMathOperator{\tr}{tr}

\DeclareMathOperator{\cyc}{cyc}
\DeclareMathOperator{\cont}{cont}

\DeclareMathOperator{\ch}{ch}
\DeclareMathOperator{\ev}{ev}
\DeclareMathOperator{\GL}{GL}
\newcommand{\id}{\mathrm{Id}}
\newcommand{\arxiv}[1]{\href{http://arxiv.org/abs/#1}{arXiv:{#1}}}

\DeclarePairedDelimiter{\abs}{|}{|}
\DeclarePairedDelimiter{\set}{\{}{\}}
\DeclarePairedDelimiter{\multiset}{\{}{\}}
\DeclarePairedDelimiter{\bra}{\langle}{|}
\DeclarePairedDelimiter{\ket}{|}{\rangle}
\DeclarePairedDelimiter{\no}{:}{:}
\DeclarePairedDelimiter{\paren}{(}{)}

\def\be{\begin{equation}}
\def\ee{\end{equation}}

\def\bea{\begin{eqnarray}}
\def\eea{\end{eqnarray}}

\def\bt{\begin{theorem}}
\def\et{\end{theorem}}

\def\ra{\rightarrow}
\def\ss{\subset}
\def\deq{\coloneqq}

\def\br{\begin{remark}\small}
\def\er{\end{remark}}

\def\&{&{\hskip -20pt}}

\def\AA{\mathcal{A}}

\def\FF{\mathcal{F}}
\def\JJ{\mathcal{J}}

\def\II{\mathcal{I}}

\def\OO{\mathcal{O}}

\def\ZZ{\mathcal{Z}}

\def\Cb{\mathbf{C}}
\def\Ib{\mathbf{I}}

\def\Nb{\mathbf{N}}
\def\Pb{\mathbf{P}}

\def\Zb{\mathbf{Z}}

\def\sb{\mathbf{s}}
\def\tb{\mathbf{t}}

\def\xb{\mathbf{x}}
\def\yb{\mathbf{y}}
\def\zb{\mathbf{z }}

\def\grF{\mathfrak{F}}

\begin{document}
\baselineskip 16pt
\medskip
\begin{center}
\begin{Large}\fontfamily{cmss}
\fontsize{17pt}{27pt}
\selectfont
\textbf{2D Toda $\tau$-functions as combinatorial generating functions}\footnote{Work supported by the Natural Sciences and Engineering Research Council of Canada (NSERC) and the Fonds de recherche du Qu\'ebec -- Nature et technologies (FRQNT).
}
\end{Large}\\
\bigskip
\begin{large} {Mathieu Guay-Paquet}$^{1}$ and {J. Harnad}$^{2,3}$
 \end{large}
\\
\bigskip
\begin{small}
$^{1}${ \em Universit\'e du Qu\'ebec \`a Montr\'eal\\
201 Av du Pr\'esident-Kennedy,
Montr\'eal QC, \ Canada H2X~3Y7 \\
email: mathieu.guaypaquet@lacim.ca} \\
\smallskip
$^{2}${ \em Centre de recherches math\'ematiques,
Universit\'e de Montr\'eal\\ C.~P.~6128, succ. centre ville,
Montr\'eal,
QC, Canada H3C 3J7 \\ e-mail: harnad@crm.umontreal.ca} \\
\smallskip
$^{3}${ \em Department of Mathematics and
Statistics, Concordia University\\ 7141 Sherbrooke W.,
Montr\'eal, QC
Canada H4B 1R6} \\
\end{small}
\end{center}
\bigskip

\begin{abstract}
Two methods of constructing 2D Toda $\tau$-functions that are generating functions for certain geometrical invariants of a combinatorial nature are related. The first involves generation of paths in the Cayley graph of the symmetric group $S_n$ by multiplication of the conjugacy class sums $C_\lambda \in \Cb[S_n]$ in the group algebra by elements of an abelian group of central elements. Extending the characteristic map to the tensor product $\Cb[S_n]\otimes\Cb[S_n]$ leads to double expansions in terms of power sum symmetric functions, in which the coefficients count the number of such paths. Applying the same map to sums over the orthogonal idempotents leads to diagonal double Schur function expansions that are identified as $\tau$-functions of hypergeometric type. The second method is the standard construction of $\tau$-functions as vacuum state matrix elements of products of vertex operators in a fermionic Fock space with elements of the abelian group of \emph{convolution symmetries}. A homomorphism between these two group actions is derived and shown to be intertwined by the characteristic map composed with fermionization. Applications include Okounkov's generating function for double Hurwitz numbers, which count branched coverings of the Riemann sphere with nonminimal branching at two points, and various analogous combinatorial counting functions.
\end{abstract}

\section{Introduction}

Many of the known generating functions for various combinatorial invariants related to Riemann surfaces have been shown to be KP $\tau$-functions, and hence to satisfy the infinite set of Hirota bilinear equations defining the KP hierarchy, or some reduction thereof.
These include the Kontsevich matrix integral \cite{Ko}, which is a KdV $\tau$-function, the generator for Hodge invariants \cite{Ka}, the matrix integrals that generate single Hurwitz numbers \cite{P, MoSh, BEMS, EO}, and the ones for Belyi curves 
and {\em dessins d'enfants} \cite{Z, KZ, AC1}. Other generating functions are known to be $\tau$-functions of the 2D Toda hierarchy, some of which are also representable as matrix integrals. Examples are the Itzykson--Zuber 2-matrix integral \cite{IZ, ZZ}, which generates the enumeration of ribbon graphs, Okounkov's generating function for double Hurwitz numbers, counting branched covers of the Riemann sphere with fixed nonminimal branching at a pair of specified points \cite{Ok}, and the Harish-Chandra--Itzykson--Zuber (HCIZ) integral \cite{HC, IZ}, which generates the monotone double Hurwitz numbers \cite{GGN1}.

The purpose of this work is to relate two different methods of constructing 2D Toda $\tau$-functions \cite{Ta, Takeb, UTa} as generating functions for geometrical-topological invariants that have combinatorial interpretations involving counting of paths in the symmetric group.
These include the double Hurwitz numbers \cite{Ok}, which may be viewed equivalently as counting paths in the Cayley graph from one conjugacy class to another, the monotone double Hurwitz numbers \cite{GGN1}, generated by the HCIZ integral in the $N \to \infty$ limit, which count weakly monotone paths, and the mixed double Hurwitz numbers \cite{GGN2}, which count a combination of both.
To these we add a new family defined by matrix integrals \cite[Appendix~A]{HO3} that are variants of the HCIZ integral, which count combinations of weakly monotone and strictly monotone paths.
In each case, the generating function can be interpreted as a $\tau$-function of the 2D Toda integrable hierarchy that is of \emph{hypergeometric  type} \cite{HO1, HO2, HO3, OS}.
The first method is based on combining Frobenius' characteristic map, from the  center $\ZZ(\Cb[S_n])$
 of the group algebra $\Cb[S_n]$ to the algebra $\Lambda$ of symmetric functions, with automorphisms of $\ZZ(\Cb[S_n])$ defined by multiplication by elements of a certain abelian group within $\ZZ(\Cb[S_n])$.
The second is based on the usual construction of $\tau$-functions \cite{HO2, HO3, OS, Takeb} as vacuum state matrix elements of products of vertex operators and operators from the Clifford group acting on a fermionic Fock space $\FF$.

Under the characteristic map, extended to $\Cb[S_n] \otimes \Cb[S_n]$, the sum $\sum_{g\in S_n} n! \, g \otimes g$ over all diagonal elements maps to the diagonal double Schur function expansion given by the Cauchy-Littlewood formula or, equivalently, to a diagonal sum of products of the power sum symmetric functions.
This may be interpreted as the restriction of the vacuum 2D Toda $\tau$-function to flow variables given by power sums.
Certain homomorphisms of the group algebra, defined by multiplication by central elements consisting of exponentials of linear combinations of power sums in the special set of commuting elements $\set{\JJ_1, \JJ_2, \JJ_3, \ldots}$ introduced by Jucys \cite{Ju} and Murphy \cite{Mu}, give rise to a ``twisting'' of the expansions in symmetric functions which, depending on the choice of the specific element, produce $\tau$-functions of hypergeometric type \cite{HO2, HO3, OS} that may be interpreted as combinatorial generating functions.
The usual way to construct $\tau$-functions of this type is by evaluating the vacuum state matrix elements with a group element that is diagonal in the standard fermionic basis. The abelian group of such diagonal elements is identified as the group $\hat{C}$ of \emph{convolution symmetries} \cite{HO4}.

In \autoref{sec:intertwining} we define a homomorphism $\II \colon \AA_P \to \hat{C}$ from the group $\AA_P$ of central elements of the form $\set{ e^{\sum_{i=0}^\infty t_i P_i} }$, where the $P_i$'s are the power sums in the Jucys-Murphy elements, to the group $\hat{C}$ of convolution symmetries. The elements of $\AA_P$ act on $\ZZ(\Cb[S_n])$ by multiplication, and are diagonal in the basis of orthogonal idempotents $\set{F_\lambda}$, labelled by partitions $\lambda$ corresponding to irreducible representations.
The elements of $\hat{C}$ act on $\FF$ and similarly are diagonal in the standard orthonormal basis $\set{\ket{\lambda; N}}$ within each charge $N$ sector $\FF_N \ss \FF$. Composing the characteristic map with the one defining the Bose--Fermi equivalence \cite{DJKM} gives an injection $\grF_n \colon \ZZ(\Cb[S_n]) \to \FF$ of the  center of the group algebra into the fermionic Fock space that maps the basis of orthogonal idempotents $\set{F_\lambda}$ to the orthonormal basis $\set{\ket{\lambda; 0}}$.
The main result, stated in \autoref{intertwining_homomorphism}, is that this map intertwines the action of the group $\AA_P$ on $\ZZ(\Cb[S_n])$ with that of $\hat{C}$ on $\FF$.

The action of $\AA_P$ on $\ZZ(\Cb[S_n])$, when expressed in another basis $\set{C_\lambda}$ consisting of the sums over elements of the conjugacy class with cycle type $\lambda$, provides combinatorial coefficients that count paths in the Cayley graph of $S_n$, starting from an element in the conjugacy class with cycle type $\lambda$ and ending on one with type $\mu$.
These are just the matrix elements of the $\AA_P$ group element in the $\set{C_\lambda}$ basis.
The image of these basis elements under the characteristic map are, up to normalization, the power sum symmetric functions $P_\lambda$.
Applying an element of $\AA_P$ to the diagonal sum $\sum_{g\in S_n} n! \, g \otimes g$ introduces a ``twist'' that is interpretable as a sum over various classes of paths in the Cayley graph.
Applying the map $\ch \otimes \ch$ to this new element provides a double sum over the power sum symmetric functions $P_\lambda([\xb]) P_\mu([\yb])$, with coefficients given by the matrix elements of $\AA_P$ which count the number of such paths or, equivalently, a double Schur function expansion of a $\tau$-function of hypergeometric type, corresponding to a specific element of the group $\hat{C}$ of convolution symmetries. This can be viewed as a method for constructing identities between double sums over the power sum symmetric functions and diagonal double Schur function expansions without involving the usual sums over irreducible characters of $S_n$.

Several examples of this construction are provided in \autoref{sec:examples}, starting with the generating function
for the double Hurwitz numbers first studied by Okounkov \cite{Ok}.
In this case, the convolution group element is given by an elliptic $\theta$-function.
In the case of weakly monotone double Hurwitz numbers, which count paths in the Cayley graph between a pair
of elements in given conjugacy classes consisting of sequences of transpositions that are weakly monotonically increasing, it corresponds to convolution with the exponential function. Choosing the expansion parameter that counts the number of steps in a path as $z=-1/N$, the resulting sequence of 2D Toda $\tau$-functions is just the large $N$ limit of the HCIZ matrix integral \cite{GGN1, GGN2}. The mixed double Hurwitz numbers, consisting of a combination of weakly monotonically increasing sequences and unordered ones,  are obtained by multiplying the two $\AA_P$ group elements.

A fourth example is introduced, in which the generating 2D Toda $\tau$-function is also interpretable as a matrix integral analogous to the HCIZ integral, but with the exponential trace product coupling replaced by a noninteger power of the characteristic polynomial of the product,
as discussed in  \cite[Appendix~A]{HO3}. This is shown to be a  generating function for  paths in the Cayley graph consisting of a sequence of weakly monotonically increasing transpositions followed by a sequence of strongly monotonically increasing ones. The identification
of the combinatorial meaning of such matrix integrals is a novel feature. This type of coupling has also been considered in the study of the spectral statistics of 2-matrix models \cite{BPF, BGS}. A final case considered here is  the family of hypergeometric $\tau$-functions introduced
recently in \cite{AMMN} as examples of  $\tau$-functions having a similar structure to the Hurwitz generating functions.
 These are shown to be generating functions for the number of multiple sequences of strictly monotonically increasing paths in the Cayley graph connecting elements in a  pair of conjugacy classes. The combinatorial significance of this family of $\tau$-unctions has never
 previously been derived. 
 
 In a sequel to this work \cite{HO5}, a broader class of generating functions is considered,
 in which the underlying ``twist'' homomorphism is generated by an arbitrary rational function, leading to Hurwitz
 numbers of multiparametric type, grouped into ``coloured'' branch points, in which the total ramification
 in each group is fixed, and a signed counting is introduced, corresponding to the parity of the number of branch
 points within each group. A further generalization consists of assigning an arbitrary $1$-parameter family of weightings 
 to the branched coverings, and to the paths in the Cayley graph. This  includes, as special cases, all previously 
 studied examples  and allows for an infinite variety of further classes of weighted Hurwitz numbers, including a natural notion of {\em quantum}
 Hurwotz numbers, in which the weighting may be related to the energy distribution of a quantum Bosonic gas \cite{GH}
 
\section[The characteristic map, twisting homomorphisms and convolution symmetries]
{The characteristic map, twisting homomorphisms \\ and convolution symmetries}

\subsection{The characteristic map and the Cauchy-Littlewood formula}

Let $\Lambda = \Cb[P_1, P_2, \ldots]$ be the ring of symmetric functions, equipped with the usual projection homomorphism
\be\begin{split}
  \ev_{n,\xb} \colon \Lambda &\to \Lambda_n \cr
  P_k &\mapsto \sum_{a=1}^n x_a^k
\end{split}\ee
onto the ring $\Lambda_n$ of symmetric polynomials in $n$ variables for each $n$.
The two bases of $\Lambda$ relevant for our purposes will be the power sum symmetric functions \cite{Mac}
\be
 P_\mu :=  P_{\mu_1} P_{\mu_2} \cdots P_{\mu_{\ell(\mu)}}
\ee
and the Schur symmetric functions $S_\lambda$, both labelled by integer partitions $\lambda = \lambda_1 \ge \dots  \lambda_{\ell(\lambda)} > 0$, 
$\mu = \mu_1 \ge \dots  \mu_{\ell(\mu)} > 0$.
These are related  by the Frobenius formula,
\bea
  P_\mu = \sum_{\mu \atop  \abs{\lambda} = \abs{\mu}} \chi_\lambda(\mu) S_\lambda,  \quad 
  S_\lambda = \sum_{\lambda\atop \abs{\mu} = \abs{\lambda} }\chi_\lambda(\mu) \frac{P_\mu}{Z_\mu},
  \label{Frobenius_formula}
\eea
where  $\chi_\lambda(\mu)$ are the irreducible characters of the symmetric groups
(with $\mu$  denoting the conjugacy class consisting of elements with cycle lengths $\mu_i$)
and, denoting the number of parts of $\lambda$ equal to $i$ by $m_i$,
\be
  Z_\mu \deq \prod_i m_i! \, i^{m_i}.
\ee

The irreducible characters $\chi_\mu$ also appear in the change of basis formula between two important bases of the  center $\{\ZZ(\Cb[S_n])\}$ of the symmetric group algebra, namely the conjugacy class sums $C_\mu$, which consist of sums of all permutation with a fixed cycle type $\mu$,
\be
C_\mu \deq \sum_{\mathclap{\substack{g \in S_n \\ \cyc(g) = \mu}}} g, 
\ee
 and the orthogonal idempotents $\{F_\lambda\}$,  corresponding to the irreducible representations of $S_n$,
which have the useful computational property that 
 \be
 F_\lambda F_\lambda = F_\lambda, \quad F_\lambda F_\nu = 0 \text{ for }\lambda \neq \mu.
 \ee
These are similarly related by
\be
  C_\mu = \frac{1}{Z_\mu} \sum_{\lambda,\, \abs{\lambda} = \abs{\mu}} h_\lambda \chi_\lambda(\mu) F_\lambda,  \quad
  F_\lambda = \frac{1}{h_\lambda} \sum_{\mu,\, \abs{\mu} = \abs{\lambda}} \chi_\lambda(\mu) C_\mu,
  \label{C_mu_F_lambda}
\ee
where $h_\lambda$ is the product of the hook lengths of the partition $\lambda$, also given by the formulae
\be
  h_\lambda^{-1}
    = \frac{\chi_\lambda(\id)}{\abs{\lambda}!}
    = \det\left(\frac{1}{(\lambda_i - i + j)!}\right) \Bigg|_{1 \le i,j \le \ell(\lambda)}.
\ee

Frobenius's characteristic map is a linear map that intertwines these changes of bases in $\ZZ(\Cb[S_n])$ and $\Lambda$, defined by
\be\begin{split}
  \ch_n \colon \ZZ(\Cb[S_n]) &\to \Lambda \\
  C_\mu  &\mapsto P_\mu  / Z_\mu  \\
  F_\lambda &\mapsto S_\lambda / h_\lambda.
\end{split}\ee
In fact this map is a linear isomorphism if we restrict its codomain to the space of homogeneous  symmetric functions of degree $n$.
It will be useful to extend the map $\ch_n$ to the whole group algebra $\Cb[S_n]$ by defining
\be
  \ch_n(g) \deq P_{\cyc(g)} / n!
\ee
for a permutation $g \in S_n$. 
Extending the tensor product map $\ch_n \otimes \ch_n$ bilinearly to the direct sum $\bigoplus_{n\in \Nb^+}\Cb[S_n] \otimes \Cb[S_n]$
\be
\label{eq:basic-element}
\ch \otimes \ch :\bigoplus_{n\in \Nb^+}\Cb[S_n] \otimes \Cb[S_n] \ra \Lambda\otimes \Lambda
\ee
then gives
\be\label{eq:ch-ch-basic}
  \ch \otimes \ch \left( \bigoplus_{n} \sum_{g \in S_n} n! \, g \otimes g \right)
    = \sum_{\mu} \frac{1}{Z_\mu} P_\mu ([\xb]) P_\mu ([\yb])
    = \sum_{\lambda} S_\lambda ([\xb]) S_\lambda ([\yb])
    = \prod_{a, b} \frac{1}{1 - x_a y_b},
   \ee
where we have identified $\Lambda \otimes \Lambda$ with the ring of symmetric functions in two infinite
sets of variables $\xb = (x_1, x_2, \dots )$ and $\yb = (y_1, y_2, \dots )$.
The last equality is just the Cauchy-Littlewood formula (\cite{Mac}).
Restricting the 2D Toda flow variables
\be
  \tb = (t_1, t_2, \ldots), \qquad
  \sb = (s_1, s_2, \ldots)
\ee
to the power sum values
\be
  t_i = \frac{1}{i} \sum_{a=1}^\infty x_a^i \qquad
  s_i = \frac{1}{i} \sum_{b=1}^\infty y_b^i
\label{power_sum_2toda_vars}
\ee
we have
\be
\prod_{a, b } \frac{1}{1 - x_a y_b} = e^{\sum_{i=1}^\infty i t_i s_i},
\ee
which is the vacuum 2D Toda $\tau$-function, restricted to the values~\eqref{power_sum_2toda_vars}.

\subsection{\texorpdfstring{``Twisting''}{"Twisting"} homomorphisms: multiplication by power sums in the Jucys-Murphy elements}
\label{sec:twisting}

The map \eqref{eq:ch-ch-basic}  can be ``twisted'' by elements of an abelian group $\AA_{P,n}$ acting on the  center $\ZZ(\Cb[S_n])$ to obtain other 2D Toda  $\tau$-functions of interest as follows.
The Jucys-Murphy elements $\set{ \JJ_b \in \Cb[S_n] }_{b=1,\ldots,n}$ are defined as sums of transpositions,
\be
  \JJ_b \deq \sum_{a=1}^{b-1} (a\,b).
\ee
They are easily seen to generate a commutative subalgebra of $\Cb[S_n]$, and any symmetric polynomial in them is in the  center $\ZZ(\Cb[S_n])$.
We can adjoin  to the ring of symmetric functions $\Lambda$ a ``trivial'' element $P_0$, taking
value $n$ under the extended evaluation map
\be\begin{split}
  \ev_{n,\JJ} \colon \Lambda[P_0] &\to \ZZ(\Cb[S_n]) \\
  G &\mapsto G(\JJ) \\
  P_k &\mapsto P_k(\JJ) \deq \sum_{b=1}^n \JJ_b^k \\
  P_0 &\mapsto P_0(\JJ) \deq n \, \id.
\end{split}\ee

\br
  While the trivial element $P_0$ acts like a scalar for any fixed $n$, it allows us to write down expressions for conjugacy classes $C_\lambda$ which hold uniformly for all $n$ (see \cite{CGS, DG}), such as
  \begin{align}
    P_1(\JJ) &= C_{21^{n-2}}, &
    P_2(\JJ) - \tfrac12 P_0(\JJ) \big( P_0(\JJ) - 1 \big) &= C_{31^{n-3}} \cr
    1 &= C_{1^n}, &
    \tfrac12 P_{11}(\JJ) - \tfrac32 P_2(\JJ) + \tfrac12 P_0(\JJ) \big( P_0(\JJ) - 1 \big) &= C_{221^{n-4}}.
  \end{align}
  From these follow equations for products of conjugacy classes such as
  \be
    C_{21^{n-2}} \cdot C_{21^{n-2}} = 3 C_{31^{n-3}} + 2 C_{221^{n-4}} + \binom{n}{2} C_{1^n}.
  \ee
  In this way, the ring $\Lambda[P_0]$ can be seen as an inverse limit of the  centers $\ZZ(\Cb[S_n])$ for all $n \in \Nb$, sometimes called the Farahat-Higman algebra \cite{FH}.
\er

Endomorphisms of $\ZZ(\Cb[S_n])$ consisting of multiplication by a central element are diagonal in the basis $\set{F_\mu}$ of orthogonal idempotents.
For elements of the form $G(\JJ)$, the result of Jucys \cite{Ju} and Murphy \cite{Mu} gives the eigenvalues as
\be
  G(\JJ) F_\lambda = G \big( \cont(\lambda) \big) F_\lambda,
  \label{G_eigenvalue}
\ee
where $\cont(\lambda)$ is the \emph{multi}set (possibly with repeated values) of contents of the boxes $(i, j)$ appearing in the Young diagram for the partition $\lambda$,
\be
  \cont(\lambda) \deq \multiset{ j-i : (i,j) \in \lambda }.
\ee
If  $G\in\Lambda[P_0]$ is expressible in the form of a  product 
\be
G=f(P_0)\prod_{a=1}^\infty F(x_a), \quad P_i =\sum_{a=1}^\infty x_a^i, 
\label{Gproductform}
\ee 
the eigenvalue $ G \big( \cont(\lambda) \big)$ is expressible as a content product:
\be
 G \big( \cont(\lambda) \big) = f(|\lambda|) \prod_{(i,j) \in \lambda} F(j-i).
\label{G_content_product}
\ee

Our ``twisting'' of the map $\ch \otimes \ch$ is defined to act on the second tensor factor only through multiplication
by a symmetric function $G(\JJ)$ for $G \in \Lambda[P_0]$ before applying the Frobenius characteristic map:
\be\label{eq:twisted-operator}
  \ch \otimes \big(\ch \circ \, G(\JJ)\big) \colon \ZZ(\Cb[S_n]) \otimes \ZZ(\Cb[S_n]) \to \Lambda \otimes \Lambda.
\ee
Using \eqref{G_eigenvalue}, it is easy to compute the result of applying the twisted homomorphism~\eqref{eq:twisted-operator} to the element \eqref{eq:basic-element} in the basis $\set{S_\lambda ([\xb]) S_\mu ([\yb])}$ in three steps.
First we apply the map $\ch$ to the left tensor factor
\be
  \ch \otimes \id \colon
    \sum_{g \in S_n} n! \, g \otimes g \mapsto     \sum_{\mathclap{\mu,\, \abs{\mu} = n}}   P_\mu ([\xb]) \otimes C_\mu =
    \sum_{\mathclap{\lambda,\, \abs{\lambda} = n}} h_\lambda S_\lambda ([\xb]) \otimes F_\lambda
\ee
by eqs.~(\ref{Frobenius_formula}), (\ref{C_mu_F_lambda}).
Then we multiply the right tensor factor  by $G(\JJ)$
\be
  \id \otimes G(\JJ) \colon
    \sum_{\mathclap{\lambda,\, \abs{\lambda} = n}} h_\lambda  S_\lambda ([\xb]) \otimes F_\lambda \mapsto
    \sum_{\mathclap{\lambda,\, \abs{\lambda} = n}} G \big( \cont(\lambda) \big) h_\lambda S_\lambda ([\xb]) \otimes F_\lambda.
\ee
And finally, we apply the map $\ch$ to the right tensor factor
\be\label{eq:twisted-S}
  \id \otimes \ch \colon
    \sum_{\mathclap{\lambda,\, \abs{\lambda} = n}} G \big( \cont(\lambda) \big) h_\lambda S_\lambda ([\xb]) \otimes F_\lambda \mapsto
    \sum_{\mathclap{\lambda,\, \abs{\lambda} = n}} G \big( \cont(\lambda) \big) S_\lambda ([\xb]) S_\lambda ([\yb]).
\ee
As will be seen in \autoref{sec:fermionic}, this is the restriction of a 2-KP $\tau$-function of hypergeometric type to the values~\eqref{power_sum_2toda_vars} of the flow parameters which, by suitable normalization, can be extended to a $\Zb$-lattice of 2D Toda $\tau$-functions.

We can perform the same computation in the basis $\set{P_\lambda ([\xb]) P_\mu ([\yb])}$ instead.
Multiplying the basis elements $C_\lambda$ by $G(\JJ)$ gives a linear combination
\be
  G(\JJ) C_\lambda = \sum_{\mu} G_{\lambda \mu} C_\mu,
\ee
where the coefficients $G_{\lambda \mu}$ are given in general by the character sum
\be
  G_{\lambda \mu} = \frac{1}{Z_\lambda} \sum_\nu G \big( \cont(\nu) \big) \chi_\nu(\lambda) \chi_\nu(\mu).
\ee
As will be seen below, in many cases $G_{\lambda \mu}$ is a combinatorial number, counting certain types of paths in the Cayley graph of $S_n$ from an element in the conjugacy class of type $C_\lambda$ to one in the class $C_\mu$.
Applying the twisted homomorphism~\eqref{eq:twisted-operator} to the element~\eqref{eq:basic-element} in three steps again gives
\begin{align}
  \ch \otimes \id &\colon
    \sum_{g \in S_n} n! \, g \otimes g \mapsto
    \sum_{\mathclap{\lambda,\, \abs{\lambda} = n}} P_\lambda ([\xb]) \otimes C_\lambda \\
  \id \otimes G(\JJ) &\colon
    \sum_{\mathclap{\lambda,\, \abs{\lambda} = n}} P_\lambda ([\xb]) \otimes C_\lambda \mapsto
    \sum_{\mathclap{\lambda,\, \mu,\, \abs{\lambda} = \abs{\mu} = n}} G_{\lambda \mu}  P_\lambda ([\xb]) \otimes C_\mu \\
  \label{eq:twisted-P}
  \id \otimes \ch &\colon
    \sum_{\mathclap{\lambda,\, \mu,\, \abs{\lambda} = \abs{\mu} = n}} G_{\lambda \mu} P_\lambda ([\xb]) \otimes C_\mu \mapsto
    \sum_{\mathclap{\lambda,\, \mu,\, \abs{\lambda} = \abs{\mu} = n}} z_\mu^{-1} G_{\lambda \mu} P_\lambda ([\xb]) P_\mu ([\yb]).
\end{align}
Comparing~\eqref{eq:twisted-S} and~\eqref{eq:twisted-P}, we get a twisted version of~\eqref{eq:ch-ch-basic}:
\be\label{eq:ch-ch-twisted}
  \ch \otimes \ch \left( \sum_{g \in S_n} n! \, g \otimes \big( G(\JJ) g \big) \right)
    = \sum_{\mathclap{\substack{\lambda,\, \mu \\ \abs{\lambda} = \abs{\mu} = n}}} G_{\lambda \mu} P_\lambda ([\xb]) P_\mu ([\yb])
    = \sum_{\mathclap{\lambda,\, \abs{\lambda} = n}} G \big( \cont(\lambda) \big) S_\lambda ([\xb]) S_\lambda ([\yb]).
\ee

\subsection{Interpretation as generating functions}
\label{sec:generating_functions}

We now consider the combinatorial meaning of the coefficients $G_{\lambda \mu}$.
If the operator $G(\JJ)$ is taken to be the power series in a formal parameter $z$ given by
\be
  G_c(z, \JJ)
    \deq e^{P_1(\JJ) z}
    = \sum_{k = 0}^\infty P_1(\JJ)^k \frac{z^k}{k!},
\ee
then the coefficient of $z^k / k!$ in $G_c(z, \JJ)$ is the element
\be
P_1(\JJ)^k = \left( \sum_{a < b \atop b =1}^n (a b) \right)^k=(C_{(2, 1^{n-2})})^k.
\ee
This acts on the group algebra $\Cb[S_n]$ by multiplication by every possible product
\be\label{eq:trans-prod}
  (a_1 \, b_1) (a_2 \, b_2) \cdots (a_k \, b_k)
\ee
of $k$ (not necessarily disjoint, nor even distinct) transpositions.
Thus, for any pair of permutations $g, h \in S_n$, the coefficient of $g \otimes h \, z^k / k!$ in the element
\be
  \sum_{g \in S_n} n! \, g \otimes \big( G_c(z, \JJ) g \big)
\ee
is the number of solutions in $S_n$ of the equation
\be
  h = (a_1 \, b_1) (a_2 \, b_2) \cdots (a_k \, b_k) g,
\ee
which is precisely the number of $k$-step walks from the vertex $g$ to the vertex $h$ in the Cayley graph of $S_n$ generated by all transpositions.
If we then apply the characteristic map $\ch \otimes \ch$ to this element, as in~\eqref{eq:ch-ch-twisted}, we see that the coefficient $G_{\lambda \mu}$ in this case is the generating function for $k$-step walks in the Cayley graph from any vertex $g$ with cycle type $\lambda$ to any vertex $h$ with cycle type $\mu$.

As another example, take the operator $G$ to be the generating function $H(z)$  for the 
complete symmetric functions. Then  $G(\JJ)$ is the power series
\be\label{eq:monotone-definition}
  H(z, \JJ)
    \deq \prod_{b=1}^n \frac{1}{1 - z \JJ_b}
    = \sum_{k=0}^\infty z^k \sum_{b_1 \leq b_2 \leq \cdots \leq b_k} \JJ_{b_1} \JJ_{b_2} \cdots \JJ_{b_k}.
\ee
The eigenvalue of this operator acting on the basis elements $F_\lambda$ is given
by
\be
 H(z, \JJ) F_\lambda = r_\lambda^H(z) F_\lambda
 \ee
 where
 \be
  r_\lambda^H(z) =\prod_{(ij) \in \lambda} (1-z(j-i))^{-1}.
  \ee
The coefficient of $z^k$ in $H(z, \JJ)$ is the operator on $\Cb[S_n]$ which acts by multiplication by every possible product~\eqref{eq:trans-prod} subject to the restriction that
\be
  b_1 \leq b_2 \leq \cdots \leq b_k,
\ee
where $a_i < b_i$ by convention.
The corresponding walks in the Cayley graph are called \emph{(weakly) monotone} walks, and for this choice of operator $G(\JJ)$, the coefficient $G_{\lambda \mu}$ in~\eqref{eq:ch-ch-twisted} is the generating function for $k$-step weakly monotone walks in the Cayley graph from any permutation with cycle type $\lambda$ to any permutation with cycle type $\mu$.
These are precisely the (nonconnected) monotone double Hurwitz numbers \cite{GGN1}.

As a final example we can choose $G$ to be the generating function $E(z)$ of the elementary
symmetric functions  to obtain an operator $G(\JJ)$ with  combinatorial meaning:
\be
  E(w, \JJ)
    \deq \prod_{a=1}^n (1 + w \JJ_a)
    = \sum_{k=0}^\infty w^k \sum_{b_1 < b_2 < \cdots < b_k} \JJ_{b_1} \JJ_{b_2} \cdots \JJ_{b_k},
   \label{strongly_monotone}
\ee
The eigenvalue of $ E(w, \JJ)$ acting on the basis elements $F_\lambda$ is given by
\be
 E(w, \JJ) F_\lambda = r_\lambda^E(z) F_\lambda
 \ee
 where
 \be
  r_\lambda^E(w) =\prod_{(ij \in \lambda} (1+w(j-i)).
  \ee
The inner summation in (\ref{strongly_monotone}) is now over strictly increasing sequences of $b_i$'s instead of weakly increasing sequences.
The corresponding walks in the Cayley graph are called \emph{strictly monotone} walks, and the coefficient $G_{\lambda \mu}$ becomes the generating function for these walks.

\subsection{Fermionic construction of 2-Toda \texorpdfstring{$\tau$}{tau}-functions}
\label{sec:fermionic}

In the following, $\FF$ denotes the full Fermionic Fock space, $\FF_N$ the charge $N$ sector, $N \in \Zb$, with orthonormal basis elements $\set{\ket{\lambda; N}}$ labelled by partitions $\lambda$.
The vacuum vector in the $\FF_N$ sector is denoted $\ket{N} \deq \ket{0; N}$.
The Fermi creation and annihilation operators, $\psi_i, \psi_i^\dag$ satisfy the usual anticommutation relations
\be
  [ \psi_i, \psi_j^\dag ]_+ = \delta_{ij}, \quad
  [ \psi_i, \psi_j ]_+ = 0,  \quad
  [ \psi_i^\dag, \psi_j^\dag ]_+ = 0,
\ee
and the vanishing relations
\be
  \psi_j \ket{N} = 0 \text{ for $j \leq N-1$},  \qquad
  \psi_j^\dag \ket{N} = 0 \text{ for $j \geq N$}.
\ee
The normal ordered product $\no{\OO_1\OO_2 \cdots \OO_k}$ of Fermionic operators is defined so that their matrix elements in the vacuum state $\ket{0}$ vanish.
The KP or 2D Toda flow parameters are denoted $\tb = (t_1, t_2, \ldots)$ and $\sb = (s_1, s_2, \ldots)$ and
\be
  \tb \deq [A], \qquad
  t_i \deq \frac{1}{i} \tr(A^i)
\ee
denotes their specialization to the trace invariants of a matrix $A$.
The vertex operators generating the KP and 2D Toda flows are defined as
\be
  \hat{\gamma}_{\pm}(\tb) \deq e^{\sum_{i=1}^\infty t_i J_{\pm i}},
\ee
where
\be
  J_i \deq \sum_{j \in \Zb} \no{\psi_j \psi_{j+i}^\dag}.
\ee
More generally,
\be
  \hat{g} = e^{\sum_{i, j \in \Zb} A_{ij} \no{\psi_i \psi_j^\dag}}
\ee
denotes the $\GL(\infty)$ group element determining a $\Zb$-lattice of $\tau$-functions as vacuum expectation values
\begin{align}
  \tau^{\text{KP}}_g (N, \tb) &= \bra{N} \hat{\gamma}_+ (\tb) \hat{g} \ket{N}, \\
  \tau^{\text{2D Toda}}_g (N, \tb, \sb) &= \bra{N} \hat{\gamma}_+ (\tb) \hat{g} \, \hat{\gamma}_- (\sb) \ket{N},
\end{align}

In particular, we have the abelian subgroup $\hat{C} \ss \GL(\infty)$ consisting of diagonal operators of the form
\be
  \hat{g} = \hat{C}_\rho \deq e^{\sum_{j \in \Zb} T_j \no{\psi_j \psi_j^\dag}},
\ee
where
\be
  \rho_i \deq e^{T_i}.
\ee
These are referred to as \emph{convolution symmetries} in \cite{HO4}, since in a basis consisting of monomials in a complex variable $z$, the $\rho_i$'s may be viewed as Fourier coefficients of a function $\rho(z) \in L^2(S^1)$, that acts by convolution product.
Defining $r_i$ as the ratio of consecutive elements,
\be
  r_i \deq \frac{\rho_i}{\rho_{i-1}} = e^{T_i - T_{i-1}},
\ee
we have \cite{HO4}
\be
  \hat{C}_\rho \ket{\lambda; N} = r_\lambda(N) \ket{\lambda; N},
\ee
where
\begin{align}
  r_\lambda(N) &= r_0(N) \prod_{(i,j) \in \lambda} r_{N+j-i}, \\
  r_0(N) &= \begin{cases}
    \prod_{j=0}^{N-1} \rho_j &\text{if $N>0$,} \\
    \qquad 1 &\text{if $N=0$,} \\
    \prod_{j=N}^{-1} \rho_j^{-1} &\text{if $N<0$.}
  \end{cases}
\end{align}

Since the convolution symmetry operators $\hat{C}_\rho$ are diagonal in the orthonormal basis $\ket{\lambda; N}$ and
\be
  \bra{\lambda; N} \hat{\gamma}_- \ket{0}
    = \bra{0} \hat{\gamma}_+ \ket{\lambda; N}
    = S_\lambda(\tb),
\ee
the corresponding $\tau$-functions have Schur function expansions
\begin{align}
  \tau_{C_\rho}^{\text{KP}} (N, \tb)
    &= \bra{N} \hat{\gamma}_+ (\tb) \hat{C}_\rho \ket{N}
     = \sum_\lambda r_\lambda(N) S_\lambda (\tb) \\
  \tau_{C_\rho}^{\text{2D Toda}} (N, \tb, \sb)
    &= \bra{N} \hat{\gamma}_+ (\tb) \hat{C}_\rho \hat{\gamma}_- (\sb) \ket{N}
     = \sum_\lambda r_\lambda(N) S_\lambda (\tb) S_\lambda (\sb).
\end{align}
This class of $\tau$-functions is referred to in \cite{OS} as being of \emph{hypergeometric type}, since it includes various multivariable generalizations of hypergeometric functions.

Equivalently, we may define the $N$-shifted operator
\be
  \hat{C}_\rho(N)
    \deq \hat{R}^{-N} \hat{C}_\rho \hat{R}^N
    = e^{\sum_{j \in \Zb} T_{j+N} \no{\psi_j \psi_j^\dag}},
\ee
where $\hat{R}$ is the shift operator defined by
\be
  \hat{R} \ket{\lambda; N} = \ket{\lambda; N+1}.
\ee
Then $\tau_{C_\rho}^{\text{KP}}(N, \tb)$ and $\tau_{C_\rho}^{\text{2D Toda}}(N, \tb, \sb) $ may equivalently be expressed as
\begin{align}
  \tau_{C_\rho}^{\text{KP}} (N, \tb)
    &= \bra{0} \hat{\gamma}_+ (\tb) \hat{C}_\rho(N) \ket{0} \\
  \tau_{C_\rho}^{\text{2D Toda}} (N, \tb, \sb)
    &= \bra{0} \hat{\gamma}_+ (\tb) \hat{C}_\rho(N) \hat{\gamma}_- (\sb) \ket{0}.
\end{align}

\subsection{The abelian group \texorpdfstring{$\AA_P$}{AP} and the intertwining homomorphism \texorpdfstring{$\II$}{I}}
\label{sec:intertwining}

If we choose the ``twisting'' homomorphism $G(\JJ)$ from \autoref{sec:twisting} to be the generating function $H(z, \JJ)$ for complete symmetric polynomials in Jucys-Murphy elements,  as in~\eqref{eq:monotone-definition}, it is easily verified that the eigenvalues are given by
\be
  H(z, \JJ) F_\lambda = r_\lambda^{[z]}(0) F_\lambda,
\ee
where
\begin{align}
  r_\lambda^{[z]}(0) \deq \prod_{(i,j) \in \lambda} r_{j-i}^{[z]}, &&
  r_j^{[z]} \deq \frac{1}{1 - jz}.
\end{align}
Forming a product of such elements, with the parameter $z$ replaced by a sequence of distinct values
\be
  \zb \deq \set{z_\alpha}_{\alpha = 1, \ldots, m}
\ee
and defining
\be
  \theta_i \deq \frac{1}{i} \sum_{\alpha=1}^m z_\alpha^i,
\ee
it follows that
\be
  e^{\sum_{i=1}^\infty \theta_i P_i(\JJ)} = \prod_{\alpha=1}^m H(z_\alpha, \JJ),
\ee
and hence this operator has eigenvalues
\be
  e^{\sum_{i=1}^\infty \theta_i P_i(\JJ)} F_\lambda = r_\lambda^{[\zb]}(0) F_\lambda,
\ee
where
\be\label{r_lambda_z}
  r_\lambda^{[\zb]}(0)
    \deq \prod_{\alpha=1}^m r_\lambda^{[z_\alpha]}(0)
    = \prod_{\alpha=1}^m \prod_{(i,j) \in \lambda} \frac{1}{1 - (j-i)z_\alpha}.
\ee
Extending this to include the trivial element $P_0(\JJ) = n = \abs{\lambda}$, we have
\be\label{F_lambda_eigenvector}
  e^{\sum_{i=0}^\infty \theta_i P_i(\JJ)} F_\lambda
    \deq e^{t_0 \abs{\lambda}} r_\lambda^{[\zb]}(0) F_\lambda.
\ee
Let $\AA_P$ denote the abelian group within $\overline{\Lambda[P_0]}$ consisting of elements of the form
\be\label{twist_group_element}
  e^{\sum_{i=0}^\infty \theta_i P_i} = e^{\theta_0 P_0} \prod_{\alpha=1}^m H(z_\alpha, \xb),
\ee
which acts on each  center $\ZZ(\Cb[S_n])$ via the evaluation at Jucys-Murphy elements as~\eqref{F_lambda_eigenvector}.
Applying the characteristic map $\ch \otimes \ch$ to the ``twisted'' sum corresponding to multiplication by the element~\eqref{twist_group_element} gives
\be
  \ch \otimes \ch \left( \sum_{g\in S_n} n!\, g \otimes \left( e^{\sum_{i=0}^\infty \theta_i P_i(\JJ)} g \right) \right)
    = \sum_{\lambda,\, \abs{\lambda} = n} e^{\theta_0 \abs{\lambda}} r_\lambda^{[\zb]}(0) S_\lambda ([\xb]) S_\lambda ([\yb]).
\ee

Note that, since the $\theta_i$'s may be viewed as the trace invariants of diagonal matrices having the $z_\alpha$'s as eigenvalues, the first $m$ of these $\set{\theta_1, \ldots, \theta_m}$ are independent, while the others are determined in terms of these by the solution of polynomial equations.
However, if we let $m \to \infty$ and extend $\set{z_\alpha}_{\alpha = 1, \ldots, m}$ to an infinite sequence of distinct complex parameters that avoid reciprocals of integers and satisfy the convergence property
\be
  \sum_{\alpha=1}^\infty \abs{z_\alpha} < \infty,
\ee
it follows that the infinite product
\be
  \prod_{\alpha=1}^\infty \prod_{(i,j) \in \lambda} \frac{1}{1 - (j-i)z_\alpha}
\ee
converges, and the $t_i$'s are functionally independent.

Since the image under the characteristic map $\ch$ of the  center $\ZZ(\Cb[S_n])$ is precisely the homogeneous degree $n$ part of the ring $\Lambda$ of symmetric functions, the map $\ch$ can be extended to a linear isomorphism
\be
  \ch \colon \bigoplus_{n \geq 0} \ZZ(\Cb[S_n]) \to \Lambda,
\ee
which we can compose with the Fermionization map
\be\begin{split}
  \Lambda &\to \FF_0 \\
  S_\lambda &\mapsto \ket{\lambda; 0}
\end{split}\ee
to get a linear isomorphism
\be\begin{split}
  \grF \colon \bigoplus_{n \geq 0} \ZZ(\Cb[S_n]) &\to \FF_0 \\
  F_\lambda &\mapsto \frac{1}{h_\lambda} \ket{\lambda; 0}.
\end{split}\ee
The linear action of the group $\AA_P$ on each of the summands $\ZZ(\Cb[S_n])$, extends to a diagonal action on the domain of the map $\grF$.
We also have an action of the group of convolution symmetries $\hat{C}$ on the codomain of the map $\grF$.
We now define a map $\II \colon \AA_P \to \hat{C}$ between these actions for which $\grF$ is the intertwining map.

Restricting ourselves to a set of parameters $\{z_\alpha\}_{\alpha=1, \ldots, m}$ with
\be
  \frac{1}{z_\alpha} \not \in \Zb,
\ee
we can define the homomorphism by
\be\label{A_P_homomorphism_finite}\begin{split}
  \II \colon \AA_P &\to \hat{C} \\
  e^{\sum_{i=0}^\infty \theta_i P_i}
    &\mapsto e^{\sum_{j \in \Zb} T_j \no{\psi_j \psi_j^\dag}}
    \eqqcolon \hat{C}_{\rho([\zb])},
\end{split}\ee
where
\begin{align}
  \rho_j([\zb])
    &\deq e^{T_j}
    = \begin{cases}
        e^{j \theta_0} \prod_{\alpha=1}^m \prod_{k=1}^j \frac{1}{1 - k z_\alpha}
          &\text{if $j > 0$,} \\
        e^{j \theta_0}
          &\text{if $j = 0$,} \\
        e^{j \theta_0} \prod_{\alpha=1}^m \prod_{k=j+1}^0 (1 - k z_\alpha)
          &\text{if $j < 0$,} \\
      \end{cases} \\
  T_j
    &\deq \begin{cases}
        j \theta_0 - \sum_{\alpha=1}^m \sum_{k=1}^j \ln(1 - k z_\alpha)
          &\text{if $j > 0$,} \\
        j \theta_0
          &\text{if $j = 0$,} \\
        j \theta_0 + \sum_{\alpha=1}^m \sum_{k=j+1}^0 \ln(1 - k z_\alpha)
          &\text{if $j < 0$,} \\
      \end{cases} \\
  r_j^{[\zb]}
    &= \frac{\rho_j([\zb])}{\rho_{j-1}([\zb])}
     = e^{\theta_0} \prod_{\alpha=1}^m \frac{1}{1 - j z_\alpha}.
\end{align}
It follows that
\be\label{C_rho_eigenvalue}
  \hat{C}_{\rho([\zb])} \ket{\lambda; 0}
    = e^{\theta_0 \abs{\lambda}} r_\lambda^{[\zb]}(0) \ket{\lambda; 0},
\ee
where $r_\lambda^{[\zb]}(0)$ is defined in~\eqref{r_lambda_z}.

We then have the following:
\begin{theorem}
\label{intertwining_homomorphism}
  The map $\grF \colon \bigoplus_{n \geq 0} \ZZ(\Cb[S_n]) \to \FF_0$ intertwines the multiplicative action of the group $\AA_P$ on $\bigoplus_{n \geq 0} \ZZ(\Cb[S_n])$ with the linear action of the group $\hat{C}$ on $\FF_0$ via the homomorphism $\II \colon \AA_P \to \hat{C}$.
\end{theorem}

\begin{proof}
  This follows from the fact that, up to scaling, the linear map $\grF$ takes $F_\lambda$ into $\ket{\lambda; 0}$ and these are, respectively, eigenvectors of the automorphism of $\ZZ(\Cb[S_n])$ defined by multiplication by $e^{\sum_{i=0}^\infty t_i P_i}$ and $\II \big( e^{\sum_{i=0}^\infty t_i P_i} \big)$ which, as given by~\eqref{F_lambda_eigenvector} and~\eqref{C_rho_eigenvalue}, have the same eigenvalue $ e^{t_0 \abs{\lambda}} r_\lambda^{[\zb]}(0)$.
\end{proof}

\br
  Note that on the intermediate space $\Lambda$ of the composition
  \be
    \grF \colon \bigoplus_{n \geq 0} \ZZ(\Cb[S_n]) \to \Lambda \to \FF_0,
  \ee
  multiplication by  $P_0(\JJ) \in \AA_P$ corresponds to the Eulerian operator
  \be
    \sum_{k=1}^\infty k P_k \frac{\partial}{\partial P_k} = \sum_{i=1}^\infty x_i \frac{\partial}{\partial x_i},
  \ee
  while multiplication by $P_1(\JJ) \in \AA_P$ corresponds to the cut-and-join operator of~\cite{G, GGN1, GGN2, Ka},
  \be
    \frac{1}{2} \sum_{i,j=1}^\infty \left((i+j) P_i P_j \frac{\partial}{\partial P_{i+j}} + i j P_{i+j} \frac{\partial^2}{\partial P_i \partial P_j }\right).
  \ee
\er

\br
  Alternatively, the homomorphism $\II \colon \AA_P \to \hat{C}$ may be defined by
  \be
    e^{\sum_{i=0}^\infty \theta_i P_i}
      \mapsto e^{
        \sum_{j\in \Zb}^\infty \big(
        \sum_{i=0}^\infty \theta_i
        \sum_{k=1}^j k^i \big) \no{\psi_j \psi_j^\dag}},
  \ee
  where the sum $\sum_{k=1}^j k^i$
  is defined for $j \leq 0$ by interpreting it as a polynomial in $j$ of degree $i+1$.
  Thus, multiplication by $P_i(\JJ) \in \AA_P$ corresponds, on $\hat{C}$, to the operator
  \be
    \sum_{j \in \Zb}^\infty \big(
    \sum_{k=1}^j k^i \big) \no{\psi_j \psi_j^\dag}.
  \ee
\er

\section{Examples}
\label{sec:examples}

\subsection{Double Hurwitz numbers}

Following Okounkov \cite{Ok}, for a pair of parameters $(\beta, q)$, we choose
\be
  r_j = q e^{j\beta}, \qquad
  \rho_j = q^j e^{\frac{\beta}{2} j(j+1)}, \qquad
  \tilde{\rho}_j = \rho_j q^{\frac{1}{2}} e^{\frac{3\beta}{8}}.
  \label{exp_rj_rho}
\ee
(The choice $\tilde{\rho}_j$ is used in \cite{Ok}; the choice $\rho_j$ fits more naturally with the conventions of \autoref{intertwining_homomorphism}.
For $N=0$, which is the only case needed, the two $\tau$-functions coincide.
The relationship between the two for general $N$ is indicated below.)
It follows that
\be
  r_\lambda(N) = q^{\frac{1}{2} N(N-1)} e^{\frac{\beta}{6} N(N^2-1)} q^{\abs{\lambda}} e^{3\beta N \abs{\lambda}} e^{\beta \cont_\lambda},
\ee
where
\be
  \cont_\lambda
    \deq P_1 \big( \cont(\lambda) \big)
    = \sum_{(i,j) \in \lambda} (j-i) = \frac{1}{2} \sum_{i=1}^{\ell(\lambda)} \lambda_i (\lambda_i - 2i +1),
\ee
or
\be
  \tilde{r}_\lambda(N) = r_\lambda(N) q^{\frac{N}{2}} e^{\frac{\beta N}{8}}.
\ee

Defining
\be
  \hat{F}_k \deq \frac{1}{k} \sum_{j \in \Zb} (j +1/2)^k \no{\psi_j \psi_j^\dag} \quad  \text{for  } k\in \Nb^+ \
  \text{ and}\  \hat{N} \deq \sum_{j \in \Zb} \no{\psi_j \psi^\dag_j},
  \ee
the convolution symmetry elements corresponding to $\rho$ and $\tilde{\rho}$ are
\be
  \hat{C}_\rho = q^{\hat{F}_1}  e^{\beta \hat{F}_2}  q^{-\frac{1}{2} \hat{N}} e^{-\frac{3\beta}{8} \hat{N}}, \quad 
    \hat{C}_{\tilde{\rho}} =  q^{\hat{F}_1}  e^{\beta \hat{F}_2} .
\ee
The function $\rho(z)$ with which the convolution product is taken has Fourier coefficients given by  (\ref{exp_rj_rho}).
Using the monomial basis $\{e_j := z^{-j-1}\}_{j\in \Zb}$ and summing, we obtain an elliptic $\Theta$-function
\be
\rho(z) = 
  \sum_{j \in \Zb}z^{-j-1} e^{j \ln q + \frac{\beta}{2} j(j+1)} 
   = {1\over q}\Theta\left( {1\over 2\pi i} \ln\left({z\over q}\right) +{\tau\over 2} \ \big{|} \ \tau \right),
    \ee
    where the modular parameter is
    \be
    \tau := {\beta \over 2 \pi i}.
    \ee
    
Under the homomorphism \eqref{A_P_homomorphism_finite}, the element of $\AA_P$ mapping to $\hat{C}_\rho$ is thus
\be\label{double_hurwitz_twist}
  \II \colon e^{\ln q \, P_0 + \beta P_1} \mapsto \hat{C}_\rho.
\ee
The corresponding 2D Toda $\tau$-functions are
\begin{align}
  \tau^{\text{2D Toda}}_{C_\rho}(N, \tb, \sb)
    &= \bra{N} \hat{\gamma}_+ (\tb) \hat{C}_\rho \hat{\gamma}(\sb) \ket{N} \\
  \tau^{\text{2D Toda}}_{C_{\tilde{\rho}}}(N, \tb, \sb)
    &= \bra{N} \hat{\gamma}_+ (\tb) \hat{C}_{\tilde{\rho}} \hat{\gamma}(\sb) \ket{N}
     = q^{\frac{N}{2}} e^{\frac{\beta N}{8}} \tau^{\text{2D Toda}}_{C_\rho}(N, \tb, \sb).
\end{align}
The generating function for the double Hurwitz numbers \cite{Ok}
\be
  F_{C_\rho} (\tb, \sb)
    = \sum_{n=1}^\infty q^n
      \sum_{b=0}^\infty \frac{\beta^b}{b!}
      \sum_{\mathclap{\substack{\lambda,\, \mu \\ \abs{\lambda}=\abs{\mu}=n}}}
        \operatorname{Hur}_{b} (\lambda, \mu) \, P_\lambda(\tb) P_\mu(\sb),
\ee
which counts only connected branched coverings of $\Cb\Pb^1$, is then the logarithm
\be
  F_{C_\rho} (\tb, \sb)
    = \ln\paren*{ \tau^{\text{2D Toda}}_{C_\rho}(0, \tb, \sb) }
\ee
of
\be
  \tau^{\text{2D Toda}}_{C_\rho}(0, \tb, \sb)
    = \sum_{n=1}^\infty q^n
      \sum_{b=0}^\infty \frac{\beta^b}{b!}
      \sum_{\mathclap{\substack{\lambda,\, \mu \\ \abs{\lambda}=\abs{\mu}=n}}}
        \operatorname{Cov}_{b} (\lambda, \mu) \, P_\lambda(\tb) P_\mu(\sb),
\ee
where $n$ is the number of sheets in the covering, $b$ is the number of simple branch points in the base, $\lambda$ and $\mu$ are the ramification types at $0$ and $\infty$, and $\operatorname{Cov}_{b} (\lambda, \mu)$ is the total number of such coverings.

\subsection{Monotone double Hurwitz numbers}

Consider the Harish-Chandra--Itzykson--Zuber (HCIZ) integral
\be
  \II_N(z, A, B)
    = \int_{\mathrlap{U \in U(N)}} \, e^{-zN \tr(UAU^\dag B)} d\mu(U)
    = \left(\prod_{k=0}^{N-1} k!\right) \frac{\det\big(e^{-zNa_i b_j}\big)_{1 \leq i,j \leq N}}{\Delta({\bf a}) \Delta ({\bf b})}
    \label{HCIZ}
\ee
where $d\mu(U)$ is the Haar measure on $U(N)$, $A$ and $B$ are a pair of
diagonal matrices with eigenvalues ${\bf a} = (a_1, \ldots, a_N)$, ${\bf b} = (b_1, \ldots, b_N)$ respectively, and $\Delta({\bf a})$, $\Delta({\bf b})$ are the Vandermonde determinants.
Defining
\be
  \FF_N \deq \frac{1}{N^2} \, \ln\big(\II_N(z, A, B)\big),
\ee
it was shown in \cite{GGN1} that this admits an expansion
\be
  \FF_N
    = \sum_{g=0}^\infty
      \sum_{n=0}^\infty \frac{(-z)^n}{n!}
      \sum_{\mathclap{\substack{\lambda,\, \mu \\ \abs{\lambda}=\abs{\mu}=n \\ \ell(\lambda),\, \ell(\mu) \leq N}}}
        \vec{H}_g(\lambda, \mu) \frac{P_\lambda([A]) P_\mu([B])}{(-N)^{2g + \ell(\lambda) + \ell(\mu)}},
\ee
where $\vec{H}_g(\lambda, \mu)$ is a monotone double Hurwitz number, which equals the number of transitive $r$-step monotone walks in the Cayley graph of $S_n$ from a permutation with cycle type $\lambda$ to one with cycle type $\mu$ and
\be
  r = 2g - 2 + \ell(\lambda) + \ell(\mu).
\ee
It is also well-known that the HCIZ integral $\II_N(z, A, B)$ is, within the normalization factor
\be
  r^{\exp}_0(N) \deq \frac{1}{\prod_{k=0}^{N-1} k!},
\ee
equal to the 2D Toda $\tau$-function $\tau^{\text{HCIZ}}(N, \tb, \sb)$ with double Schur function expansion 
\cite[Appendix~A]{HO3}
\be
  r^{\exp}_0(N) \II_N(z, A, B)
    = \tau^{\text{HCIZ}}(N, \tb, \sb)
    \deq \sum_{\mathclap{\substack{\lambda \\ \ell(\lambda) \leq N}}}
      r_\lambda^{\exp}(N) S_\lambda(\tb) S_\lambda(\sb),
\ee
where
\begin{align}
  r_\lambda^{\exp}(N) = \frac{(-zN)^{\abs{\lambda}}}{\left(\prod_{k=0}^{N-1}k!\right) N_{(\lambda)}}, &&
  (N)_{\lambda} = \prod_{(i,j) \in \lambda} (N+j-i),
\end{align}
evaluated at the parameter values
\be\label{t_s_A_B}
  \tb = [A], \qquad
  \sb = [B].
\ee
This may be expressed as the fermionic vacuum state expectation value
\be
  \tau_{C_{\exp}}(N, [A], [B])
    = \bra{N} \hat{\gamma}_+([A]) \hat{C}_{\exp} \hat{\gamma}_-([B]) \ket{N}
    = \bra{0} \hat{\gamma}_+([A]) \hat{C}_{\exp}(N) \hat{\gamma}_-([B]) \ket{0}
\ee
where
\begin{align}
  \hat{C}_{\exp} &\deq e^{\sum_{j=0}^\infty (j \ln(-zN) - \ln(j!)) \no{\psi_j \psi_j^\dag}} \\
  \hat{C}_{\exp}(N) &\deq e^{\sum_{j=-N}^\infty((j+N) \ln(-zN) - \ln((j+N)!)) \no{\psi_j \psi_j^\dag}}.
\end{align}
The convolution group element $\hat{C}_{\exp}(N)$ is the image, under the homomorphism $\II \colon \AA_P \to \hat{C}$, of the element
\begin{gather}
  \label{monotone_hurwitz_twist}
  e^{-\ln(-zN) P_0(\JJ)} e^{\sum_{i=1}^\infty \frac{(-1/N)^i}{i} P_i(\JJ)}
    = e^{-\ln(-zN) P_0(\JJ)} H(-1/N, \JJ) \\
  \II\Big(e^{-\ln(-zN) P_0(\JJ)} H(-1/N, \JJ) e^{\ln q \, P_0(\JJ) + \beta P_1(\JJ)} \Big) = \hat{C}_{\exp}(N).
\end{gather}

\subsection{Mixed double Hurwitz numbers}

The mixed monotone Hurwitz numbers are defined in \cite{GGN2} as the number of $r$-step walks in the Cayley graph of $S_n$ from a permutation with cycle type $\lambda$ to one with cycle type $\mu$, subject to the restriction that the first $p \leq r$ steps form a weakly monotone walk, and the last $r-p$ steps are unrestricted.
This case has a generating function that is obtained by composing the group element \eqref{double_hurwitz_twist} in $\AA_P$ corresponding to the ordinary double Hurwitz numbers with the one \eqref{monotone_hurwitz_twist} corresponding to the monotone ones.
Applying the homomorphism $\AA_P$ to the product therefore gives the product of the convolution group elements
\be
  \II \Big( e^{\ln q \, P_0(\JJ) + \beta P_1(\JJ)} e^{-\ln(-zN) P_0(\JJ)} H(-1/N, \JJ) \Big)
    = e^{\ln q \, \hat{F}_1 + \beta \hat{F}_2} \hat{C}_{\exp}(N).
\ee
It follows that the factor $r_\lambda(N)$ that enters in the double Schur function expansion of the corresponding mixed double Hurwitz number generating function is given by the product of the ones for these two cases,
\be
  r_\lambda(N)
    = q^{\frac{1}{2} N(N-1)} e^{\frac{\beta}{6} N(N^2-1)} q^{\abs{\lambda}} e^{\beta N \abs{\lambda}} e^{\beta \cont_\lambda}
      \, \frac{(-zN)^{\abs{\lambda}}}{\left(\prod_{k=0}^{N-1}k!\right) N_{(\lambda)}}.
\ee

\subsection{Determinantal matrix integrals as generating functions}

Following \cite[Appendix~A]{HO3}, we can obtain a new class of combinatorial generating functions that generalize the case of the HCIZ integral as follows.
Choose a pair $(\alpha, q)$ of (real or complex) parameters, with $\alpha$ not a positive integer, and define
\be
  \rho_j^{(\alpha, q)} \deq \begin{cases}
    \frac{q^j (1-\alpha)_j}{j!} &\text{if $j \geq 1$,} \\
    \hspace{1.4em} 1 &\text{if $j \leq 0$,}
  \end{cases}
\ee
where
\be
  (a)_j \deq a (a+1) \cdots (a+j-1)
\ee
is the (rising) Pochhammer symbol.
Then
\be
  r_j^{(\alpha, q)}
    \deq \frac{\rho_j^{(\alpha,q)}}{\rho_{j-1}^{(\alpha,q)}}
    = \begin{cases}
        \frac{q(j-\alpha)}{j} &\text{if $j \geq 1$,} \\
        \quad 1 &\text{if $j \leq 0$,}
      \end{cases}
\ee
and
\be
  T_j^{(\alpha,q)}
    \deq \ln \rho_j^{(\alpha,q)}
    = \begin{cases}
        j\ln q + \ln(1-\alpha)_j - \ln (j!) &\text{if $j \geq 1$,} \\
        \hspace{5.5em} 0 &\text{if $j \leq 0$.}
\end{cases}
\ee
For any $N \in \Nb$, let
\be
  \hat{C}_{(\alpha, q)}(N) \deq e^{\sum_{j=-N}^\infty T^{(\alpha, q)}_{j+N} \no{\psi_j \psi_j^\dag}}
\ee
be the corresponding shifted convolution symmetry group element.

We then have, for $\ell(\lambda) \leq N$,
\be\label{C_alpha_rho_eigenvector}
  \hat{C}_{(\alpha, q)}(N) \ket{\lambda;0} = r_\lambda^{(\alpha, q)}(N) \ket{\lambda; 0}
\ee
where
\be
  r_\lambda^{(\alpha, q)}(N)
    = r_0^{(\alpha, q)}(N) \prod_{(i,j) \in \lambda} r_{N+j-i}^{(\alpha,q)}
    = r_0^{(\alpha, q)}(N) q^{\abs{\lambda}} \frac{(N-\alpha)_\lambda}{(N)_\lambda},
\ee
with
\be
  r_0^{(\alpha, q)}(N)
    \deq \prod_{j=0}^{N-1} \rho^{(\alpha, q)}_j
    = q^{\frac{1}{2} N(N-1)} \prod_{j=0}^{N-1} \frac{(1-\alpha)_j}{j!}
\ee
and
\be
  (a)_\lambda \deq \prod_{i=1}^{\ell(\lambda)} (a-i+1)_{\lambda_i}
\ee
the extended Pochhammer symbol corresponding to the partition $\lambda=(\lambda_1, \ldots, \lambda_{\ell(\lambda)})$.

For $N \in \Nb^+$, we have the 2D Toda chain of $\tau$-functions
\be
  \tau_{\scriptscriptstyle C_{(\alpha, q)}} (N, \tb, \sb)
    = \bra{0} \hat{\gamma}_+(\tb) \hat{C}_{(\alpha, q)}(N) \gamma_-(\sb) \ket{0}
    = \sum_\lambda r_\lambda^{(\alpha, q)}(N) S_\lambda(\tb) S_\lambda (\sb),
\ee
evaluated at the parameter values \eqref{t_s_A_B}.
As shown in \cite{HO3}, this is just the matrix integral
\begin{align}\label{det_matrix_integral}
  \tau_{\scriptscriptstyle C_{(\alpha, q)}} (N, [A] ,[B])
    &= r_0^{(\alpha, q)}(N) \int_{\mathrlap{U \in U(N)}} \,\, \det(\Ib_N - qUAU^\dag B)^{\alpha - N} d\mu(U) \\
    &= \frac{\big(\det(1 - q a_i b_j)_{1 \leq i,j \leq N}\big)^{\alpha-1}}{\Delta({\bf a}) \Delta ({\bf b})}.
\end{align}

We now use the other construction to derive an interpretation of this as a combinatorial generating function.
Evaluating the generating function for the elementary symmetric polynomials at the Jucys-Murphy elements
\be
  E(w, \JJ) \deq \prod_{a=1}^n (1 + w \JJ_a)
\ee
defines an element of $\AA_P$.
Applying the product
\be
  \left(-\frac{qz}{w}\right)^{P_0} H(z, \JJ) E(w, \JJ)
\ee
to the orthogonal idempotent $F_\lambda$, we obtain
\be
  \left(-\frac{qz}{w}\right)^{P_0} H(z, \JJ) E(w, \JJ) F_\lambda
    = q^{\abs{\lambda}} \frac{(1/w)_\lambda}{(-1/z)_\lambda} F_\lambda.
\ee
Specializing to the values
\be
  z = -1/N, \qquad
  w = \frac{1}{N-\alpha}, \qquad
  \tb = [A], \qquad
  \sb = [B]
\ee
and choosing $\ell(\lambda) \leq N$, we obtain the same eigenvalue, within a normalization factor, as in \eqref{C_alpha_rho_eigenvector}, namely
\be
  \left(q\left(\frac{\alpha}{N}-1\right)\right)^{P_0} H(-1/N, \JJ) E(-1/(N-\alpha), \JJ) F_\lambda
    = q^{\abs{\lambda}} \frac{(N-\alpha)_\lambda}{(N)_\lambda} F_\lambda
    = \frac{r_\lambda^{(\alpha, q)}(N)}{r_0^{(\alpha, q)}(N)} F_\lambda.
\ee
Under the homomorphism $\II \colon \AA_P \to \hat{C}$, we thus have
\be
  \left(q\left(\frac{\alpha}{N}-1\right)\right)^{P_0} H(-1/N, \JJ) E(-1/( N-\alpha), \JJ)
    \mapsto \frac{\hat{C}_{(\alpha, q)}(N)}{r^{(\alpha, q)}_0(N)}.
\ee
Applying the product $q^{P_0} H(z, \JJ) E(w, \JJ)$ to the conjugacy class sum $C_\lambda$ therefore gives
\be
  q^{P_0} H(z, \JJ) E(w, \JJ) C_\lambda
    = \sum_{k,\, l=0}^\infty z^k w^l
      \sum_\mu q^{\abs{\lambda}} E_{k, l}(\lambda, \mu) C_\mu,
\ee
where, similarly to the mixed double Hurwitz numbers, $E_{k, l}(\lambda, \mu)$ is the number of $(k+l)$-step walks in the Cayley graph of $S_n$ starting at a permutation with cycle type $\lambda$ and ending at a permutation of cycle type $\mu$ which obey the condition that the first $k$ steps form a weakly monotone walk, and the last $l$ steps form a strictly monotone walk.

Applying the map $\ch \otimes \ch$ to $\sum_{g \in S_n} n! \, g \otimes \big( q^{P_0}H(z, \JJ) E(w, \JJ) g \big)$ gives
\begin{align}
  \sum_{g \in S_n} n! \, g \otimes \big( q^{P_0}H(z, \JJ) E(w, \JJ) g \big)
    &\mapsto
      \sum_{k,\, l = 0}^\infty z^k w^l
      \sum_{\mathclap{\abs{\lambda} = \abs{\mu} = n}}
        q^{\abs{\lambda}} E_{k, l}(\lambda, \mu) P_\lambda(\tb) P_\mu(\sb) \notag\\
    &=
      \sum_{\abs{\lambda}=n}
        q^{\abs{\lambda}} r_\lambda(z,w) S_\lambda(\tb) S_\lambda(\sb),
\end{align}
where
\be
  r_\lambda(z,w) \deq \prod_{(i,j) \in \lambda} \frac{1 + (j-i)w}{1 - (j-i)z}
\ee
and hence
\be
  \left(q\left(\frac{\alpha}{N} -1\right)\right)^{\abs{\lambda}} r_\lambda \paren*{-\frac{1}{N}, \frac{1}{N-\alpha}}
    = \frac{r_\lambda^{(\alpha, q)}(N)}{r_0^{(\alpha, q)}(N)}.
\ee
Therefore, in the limit $N \to \infty$, the matrix integral \eqref{det_matrix_integral} is the generating function for the number of weakly-monotonic-then-strictly-monotonic double Hurwitz numbers.

\subsection{A further example: multimonotone paths}

In a recent paper by Alexandrov {\em et al} \cite{AMMN}, a further class of functions, denoted  
$Z_{(k,m)}(s, u_1, \dots , u_m |  {\bf p}^{(1)}, \dots, {\bf p}^{(k)})$, with structure similar  to the 
generating function for Hurwitz numbers was studied.
  These depend on a set of $m+1$ parameters $(s, u_1, \dots , u_m)$,  and are expressible as  sums over  $k$-fold products of Schur functions $\prod_{i=1}^kS_\lambda ( {\bf p}^{(i)}) $, whose coefficients are products of functions of the individual parameters $(s, u_1, \dots, u_m)$,  which are themselves content products  of the  type (\ref{G_content_product}).  
  For $k=1$ or $2$,  it follows from their definition that these are  KP and 2D Toda $\tau$-functions of 
  hypergeometric type;  for  $k>2$ they have no such interpretation.
  
The $k=2$ case is defined by the double Schur function expansion 
\be
Z_{(2,m)} (s, u_1, \dots , u_m |  {\bf p}^{(1)},  {\bf p}^{(2)} ) =\sum_\lambda r_\lambda^{(s, u_1, \dots , u_m)}
S_\lambda({\bf p}^{(1)} )S_\lambda({\bf p}^{(2)}), \ee
where
\be
r_\lambda^{(s, u_1, \dots , u_m)} := s^{|\lambda|} \prod_{\alpha=1}^m  \prod_{(ij) \in \lambda} (u_\alpha+ i-j). 
\ee
and the  $k=1$ case is obtained by setting ${\bf p}^{(2)} = (1, 0, 0, \dots )$.
 Although similar in form to the simple and  double Hurwitz number generating functions, no combinatorial
 interpretation  of  these  was given in \cite{AMMN}.  The case $Z_{(2,1)}$ is just the well-known  
 HCIZ  integral (\ref{HCIZ})  treated above when evaluated at the special values (\ref{t_s_A_B}). 

The combinatorial significance of $Z_{(2,m)}$ for all $m\in \Nb^+$ is very easily understood  in  our approach.
To express this example in the notational  conventions  above, it is  convenient to define slightly different expansion
 parameters
 \be
   q := (-1)^m s \prod_{\alpha=1}^m u_a,  \quad w_\alpha := -1/u_\alpha, \quad \alpha= 1, \dots , m, 
 \ee
 and denote
 \be
\tilde{Z}_{(2, m)}(q, w_1, \dots , w_m |  {\bf p}^{(1)},  {\bf p}^{(2)} ) 
:= Z_{(2, m)}(s, u_1, \dots , u_m |  {\bf p}^{(1)},  {\bf p}^{(2)}).  
 \ee
The diagonal double Schur function expansions  for $\tilde{Z}_{(2, m)}(q, w_1, \dots , w_m |  {\bf p}^{(1)},  {\bf p}^{(2)})$
may then  be re-expressed as a double series over  products  $P_\lambda({\bf p}^{(1)})P_\mu({\bf p}^{(2)})$
of  power sum symmetric functions  via the Frobenius character formula (\ref{Frobenius_formula}), 
and  further developed as multiple Taylor series in the variables $(p, w_1, \dots , w_m)$:
 \be
 \tilde{Z}_{(2, m)}(q, w_1, \dots , w_m |  {\bf p}^{(1)},  {\bf p}^{(2)}) = 
 \sum_{\mathclap{\mathclap{n=0}}}^\infty q^n 
  \sum_{{\lambda, \mu \atop\abs{\lambda} = \abs{\mu} = n}}
  \sum_{{d_1, d_2, \dots d_m = 0}}^\infty \left(  \prod_{\alpha=1}^m w_\alpha^{d_\alpha}   \right)    E^{(n, d_1, \dots d_m)}(\lambda, \mu)    P_\lambda( {\bf p}^{(1)}) P_\mu( {\bf p}^{(2)}).
  \label{multi_monotone_generating_function}
 \ee
 
The coefficients  $E^{(n, d_1, \dots d_m)}(\lambda, \mu)$   in this series
 have a simple combinatorial meaning. They are the number of paths in the Cayley graph 
 of $S_n$  generated by transpositions $(ab)$, $a<b$,   starting  from an element in the  class sum $C_\lambda$ and ending at one in the class sum  $C_\mu$, related  by multiplication by   a  product of transpositions of the form  
 \be
 (a_1b_1) \cdots (a_d  b_d ), \quad d:= \sum_{\alpha=1}^m d_i, 
 \ee
   in which the $b_i$'s are strictly  monotonically increasing within each successive segment  of length $d_i$,  
   starting at $(a_1 b_1)$.  
   
To see this, just note that the reparametrized content product 
   $\tilde{r}_\lambda^{(q, w_1, \dots , w_m)}$ appearing in the diagonal double Schur function expansion
\be
\tilde{Z}_{(2,n)}(q, w_1, \dots , w_m |  {\bf p}^{(i)}) = \sum_\lambda \tilde{r}_\lambda^{(q, w_1, \dots , w_m)}
S_\lambda({\bf p}^{(1)} )S_\lambda({\bf p}^{(2)}), 
\ee
is 
\be
\tilde{r}_\lambda^{(q, w_1, \dots , w_m)} = q^{|\lambda|}\prod_{a=1}^m  \prod_{(ij) \in \lambda} (1+w_\alpha (j-i)). 
\label{multi_monotone_eigenvalue}
\ee
From the discussion in \autoref{sec:generating_functions},  this is just the product of  the eigenvalues of the generating
functions of the elementary symmetric functions,  expressed in terms of the Jucys-Murphy elements, 
\be
E(w_\alpha, \JJ)= \prod_{a=1}^n(1+w_\alpha \JJ_a), 
\ee
and each of  these generates strictly monotonic paths. The element $G(\JJ)  \in \Lambda[P_0]$  used to define 
the  ``twist'' in this case  is therefore the product
$q^{P_0(\JJ)}\prod_{\alpha=1}^m E(w_\alpha, \JJ)$,
whose eigenvalues $\tilde{r}_\lambda^{(q, w_1, \dots , w_m)} $  in the $F_\lambda$ basis
\be
q^{P_0(\JJ)}\prod_{\alpha=1}^m E(w_\alpha, \JJ) F_\lambda =  \tilde{r}_\lambda^{(q, w_1, \dots , w_m)} F_\lambda
\ee
 are given by (\ref{multi_monotone_eigenvalue}).
The multimonotone Cayley path interpretation follows from the discussion of 
the last example in \autoref{sec:generating_functions}.

 \bigskip
\noindent {\em  Bibliographical update.}  This paper was  posted as  \arxiv{1405.6303} in May 2014
and submitted at that time to {\em Lett.~Math.~Phys.~}for publication. No substantive changes have been made 
since then,  but in the intervening time several further papers have appeared on related matters
 \cite{AMMN, H1, H2, H3, AC1, AC2,  NOr1, NOr2},  some of which have since been published.  Not all recent
 contributions give due reference to the present work,   but for the sake of  completeness, 
 we include mention of all known related works, whether they  appeared prior to,  or subsequent to the present one.


\bigskip

\end{document}